\def\bu{\bullet}
\def\marker{\>\hbox{${\vcenter{\vbox{
    \hrule height 0.4pt\hbox{\vrule width 0.4pt height 6pt
    \kern6pt\vrule width 0.4pt}\hrule height 0.4pt}}}$}\>}
\def\gpic#1{#1
     \smallskip\par\noindent{\centerline{\box\graph}} \medskip}
\documentclass[12pt]{article}
\usepackage{array,amsmath,amssymb,amsthm}  				    
\usepackage{fullpage,lineno}
\usepackage{algorithmic,algorithm}

\usepackage{color}

\begin{document}

\newtheorem{theorem}{Theorem}[section]
\newtheorem{lemma}[theorem]{Lemma}
\newtheorem{corollary}[theorem]{Corollary}
\newtheorem{prop}[theorem]{Proposition}
\newtheorem{conj}[theorem]{Conjecture}
\newtheorem{claim}[theorem]{Claim}
\theoremstyle{definition}
\newtheorem{definition}[theorem]{Definition}
\newtheorem{remark}[theorem]{Remark}
\newtheorem{alg}[theorem]{Algorithm}
\newtheorem{obs}[theorem]{Observation}
\def\qed{\ifhmode\unskip\nobreak\hfill$\Box$\bigskip\fi \ifmmode\eqno{Box}\fi}
\newcommand{\diamondplus}{$\kern0.15em\Diamond$\kern-0.73em\raisebox{0.08ex}{$+$}}

\def\nul{\varnothing} 
\def\st{\colon\,}   
\def\VEC#1#2#3{#1_{#2},\ldots,#1_{#3}}
\def\VECOP#1#2#3#4{#1_{#2}#4\cdots #4 #1_{#3}}
\def\SE#1#2#3{\sum_{#1=#2}^{#3}} 
\def\PE#1#2#3{\prod_{#1=#2}^{#3}}
\def\UE#1#2#3{\bigcup_{#1=#2}^{#3}}
\def\CH#1#2{\binom{#1}{#2}} 
\def\FR#1#2{\frac{#1}{#2}}
\def\FL#1{\left\lfloor{#1}\right\rfloor} \def\FFR#1#2{\FL{\frac{#1}{#2}}}
\def\CL#1{\left\lceil{#1}\right\rceil}   \def\CFR#1#2{\CL{\frac{#1}{#2}}}
\def\Gb{\overline{G}}
\def\NN{{\mathbb N}} \def\ZZ{{\mathbb Z}} \def\QQ{{\mathbb Q}}
\def\RR{{\mathbb R}} \def\GG{{\mathbb G}} \def\FF{{\mathbb F}}

\def\B#1{{\bf #1}}      \def\R#1{{\rm #1}}
\def\I#1{{\it #1}}      \def\c#1{{\cal #1}}
\def\C#1{\left | #1 \right |}    
\def\P#1{\left ( #1 \right )}    
\def\ov#1{\overline{#1}}        \def\un#1{\underline{#1}}

\def\la{\langle}
\def\ra{\rangle}
\def\hT{\hat T}
\def\hH{\hat H}
\def\esub{\subseteq}


\title{Spanning Trees in $2$-trees}

\author{
P. Renjith\footnotemark[1],\quad
N. Sadagopan\footnotemark[1],\quad
Douglas B. West\footnotemark[2]
}

\footnotetext[1]{Indian Institute of Information Technology, 
Design, and Manufacturing, Kancheepuram, Chennai, India.}
\footnotetext[2]{Departments of Mathematics, Zhejiang Normal University,
Jinhua, Zhejiang, China, and University of Illinois, Urbana, Illinois, USA.
Research supported by Recruitment Program of Foreign Experts,
1000 Talent Plan, State Administration of Foreign Experts Affairs, China.
}

\date{\today}
\maketitle

\begin{abstract}
A spanning tree of a graph $G$ is a connected acyclic spanning subgraph of $G$.
We consider enumeration of spanning trees when $G$ is a $2$-tree, meaning that
$G$ is obtained from one edge by iteratively adding a vertex whose neighborhood
consists of two adjacent vertices.  We use this construction order both to
inductively list the spanning trees without repetition and to give bounds on
the number of them.  We determine the $n$-vertex $2$-trees having the most and
the fewest spanning trees.  The $2$-tree with the fewest is unique; it has
$n-2$ vertices of degree $2$ and has $n2^{n-3}$ spanning trees.  Those with the
most are all those having exactly two vertices of degree $2$, and their number
of spanning trees is the Fibonacci number $F_{2n-2}$.
\end{abstract}

\baselineskip 16pt

\section{Introduction}
The problem of counting spanning trees in a graph has attracted much attention
\cite{gabow,Ramesh,Kapoor,Matsui,read,Shioura,uno,Yan}.  This number can be
computed efficiently using the Matrix Tree Theorem~\cite{Kirchhoff}.  By
``enumeration'' we instead mean the listing of all spanning trees.  We consider
the special case of $2$-trees, which form a well-known subclass of the family
of chordal graphs.  A {\it $2$-tree} is a graph obtained from the graph 
consisting of two adjacent vertices by iteratively adding one new vertex whose
neighborhood consists of two adjacent vertices.

For general graphs, the tree-listing algorithms in the literature use 
backtracking or generate trees from fundamental cycles~\cite{gabow,read}.
For $n$-vertex $2$-trees, we show that the number of spanning trees is always
between $2^{n-2}$ and $3^{n-2}$, so it is not possible to list them in
polynomial time.  Nevertheless, the structural description of $2$-trees yields
a simple sequential algorithm that lists all spanning trees without repetition.
This approach is simpler than that in \cite{gabow,read}, but the domain of
$2$-trees is much simpler than the contexts of those papers.  Given a vertex
$v$ of degree $2$ in a $2$-tree $G$, we list the spanning trees in $G$ in terms
of the spanning trees in $G-v$, separately those in which $v$ is a leaf and
those in which $v$ is not a leaf. 

This algorithm, presented in Section 2, immediately yields lower and upper
bounds $2^{n-2}$ and $3^{n-2}$ on the number of spanning trees in an $n$-vertex
$2$-tree.  We improve both bounds to obtain optimal values and characterize the
$n$-vertex $2$-trees achieving equality.  We provide the constructions in
Section 3 and prove optimality in Sections 4 and 5.

The minimum number of spanning trees in an $n$-vertex $2$-tree is $n2^{n-3}$,
achieved uniquely by the $2$-tree having $n-2$ vertices of degree $2$.
This $2$-tree is obtained from the graph with two vertices and one edge by
iteratively adding vertices of degree $2$ adjacent to the two original
vertices; it is sometimes called the {\it $n$-book}, denoted $B_n$. 

The $n$-vertex $2$-trees having the most spanning trees are those having
exactly two vertices of degree $2$.  Somewhat surprisingly, all such $2$-trees
have the same number of spanning trees.  The value is asymptotic to
$\frac{1}{\sqrt{5}}(2.618)^{n-2}$; the exact value is the Fibonacci number
$F_{2n-2}$, where $F_0=0$, $F_1=1$, and $F_m=F_{m-1}+F_{m-2}$ for $m\ge2$.
Such $2$-trees include the graphs obtained from an $n$-vertex path by
(1) making one endpoint adjacent to all the other vertices or (2) adding
edges joining vertices separated by distance $2$ along the path. 

We use standard graph-theoretic notation.  In particular, $V(G)$ and $E(G)$
denote the vertex set and edge set of a graph $G$.  The {\it neighborhood}
of a vertex $v$ is the set $N(v)$ of vertices adjacent to $v$, that is, the set
$\{u\in V(G)\colon\, uv\in E(G)\}$.  The {\it complete graph} $K_n$ is the
$n$-vertex graph in which any two vertices are adjacent, and a {\it clique} is
a set of pairwise adjacent vertices.  The subgraph of a graph $G$ {\it induced}
by a vertex subset $S$, written $G[S]$, is the subgraph defined by $V(G[S])=S$
and $E(G[S])=\{uv\in E(G)\colon\, u,v\in S\}$.  When $v\in V(G)$, we write
$G-v$ for the graph $G[V(G)-\{v\}]$. 

A vertex in a graph is {\it simplicial} if its neighborhood is a clique.
A {\it simplicial elimination ordering} of a graph $G$ is an ordering
$(v_n,\ldots,v_1)$ of $V(G)$ such that each $v_i$ is a simplicial vertex of the
induced subgraph $G[\{v_i,\ldots,v_1\}]$.  Traditionally, simplical elimination
orderings have also been called ``perfect elimination orderings''
(see~\cite{golumbic}), but there are now many types of elimination orderings,
so using the word ``simplicial'' is more informative.  Our definition of
$2$-trees explicitly constructs them in the reverse of a simplicial elimination
ordering.

A graph is {\it chordal} if no induced subgraph is a cycle of length at least
$4$.  It was proved by Dirac~\cite{dirac61} that a graph is chordal if and
only if it has a simplicial elimination ordering.  Thus the $2$-trees form a
special class of chordal graphs.  They are the connected chordal graphs having
a simplicial elimination ordering in which except for the last two vertices,
all vertices have degree $2$ when deleted.  We thus call a simplicial
elimination ordering of a $2$-tree a {\it $2$-simplicial ordering}.  Note that
$2$-trees have no cut-vertices, and hence a vertex in a $2$-tree is simplicial
if and only if it has degree $2$. 

\section{Enumeration without Repetition}

Given a $2$-tree $G$ and an associated $2$-simplicial ordering $(\VEC vn1)$,
we present an algorithm to list all spanning trees in $G$ without repetition.
Finding a $2$-simplicial ordering is trivial; a $2$-tree
with more than two vertices has a simplicial vertex of degree $2$, and every
vertex of degree $2$ is simplicial and can be used as the next vertex in the
ordering.  Hence maintaining the vertex degrees makes it very easy to
generate a $2$-simplicial ordering.

Given a $2$-tree $G$ with $2$-simplicial ordering $(v_n,\ldots,v_1)$, let
$G_i=G[v_i,\ldots,v_1]$.  Starting with $G_2$, which is isomorphic to $K_2$,
we iteratively obtain the spanning trees of $G_2,\ldots,G_n$ in order.
In particular, from the spanning trees of $G_{n-1}$ we obtain the spanning
trees of $G_n$, by listing first those in which $v_n$ is a leaf and then
those in which $v_n$ is not a leaf.  The justification of our algorithm is
based on the following observation.

\begin{obs}\label{obs1}
Let $G$ be a $2$-tree in which $v$ is a simplicial vertex with neighborhood
$\{x,y\}$.  A subgraph $T$ of $G$ is a spanning tree in which $v$ is not a leaf
if and only if $v$ has degree $2$ in $T$, the edge $xy$ is not in $T$, and the
spanning subgraph $T'$ of $G-v$ with $E(T')=E(T)\cup\{xy\}-\{vx,vy\}$ is a
spanning tree in $G-v$.
\end{obs}

\begin{algorithm}
\caption{{\em Enumerate-Spanning-trees-without-repetitions($G$)}}
\label{mainalgo}
\begin{algorithmic}[1]
\STATE{{\tt Input:}  A 2-tree G with 2-simplicial ordering $(v_n,\ldots,v_1)$, for $n\ge2$;
    define $G_i=G[\{v_i,\ldots,v_1\}]$ for $1\le i\le n$.}
 \STATE{ Initialize LIST with the single spanning tree of $G_2$.}
 \STATE{Let $\{x,y\}$ be the neighborhood of $v_i$ among $\{v_{i-1},\ldots,v_1\}$.}
 \FOR{ each spanning tree $T$ of $G_{i-1}$ in LIST,}
 \STATE{Append to LIST the two spanning trees of $G_i$ obtained by adding the
    edges $v_ix$ or $v_iy$ to $T$.}
 \STATE{ If $xy\in E(T)$, then Append to LIST the spanning tree of $G_i$ obtained
    by replacing $xy$ with the edges $v_ix$ and $v_iy$.}
    \ENDFOR
\end{algorithmic}
\end{algorithm}

\begin{theorem}\label{listall}
For a 2-tree $G$, Algorithm \ref{mainalgo} lists all spanning trees of $G$
without repetition.
\end{theorem}
\begin{proof}
Let $G$ have $n$ vertices.  We use induction on $n$, with trivial basis. 
Assume that the algorithm lists the spanning trees of $G_{n-1}$ without
repetition.  Deleting $v_n$ from any spanning tree of $G_n$ where $v_n$ is a
leaf yields a spanning tree of $G_{n-1}$, which in the algorithm produces it in
Step~5.  By Observation~\ref{obs1}, any spanning tree in which $v_n$ is not a
leaf is produced by the algorithm in Step~6.  Thus the algorithm lists all
spanning trees of $G$ without repetition.  Furthermore, since the trees are
built by iteratively adding one vertex, the time taken is proportional to $nt$,
where $t$ is the number of spanning trees.
\end{proof}

\section{Bounds and Constructions} \label{secbound}
Let $T(G)$ denote the number of spanning trees in a $2$-tree $G$.
Again let $G_i=G[\{v_i,\ldots,v_1\}]$, where $(v_n,\ldots,v_1)$ is a
$2$-simplicial ordering of $G$.

\begin{theorem}\label{lubound}
Fix $n\in\NN$ with $n\ge2$.  If $G$ is an $n$-vertex $2$-tree, then
$2^{n-2}\le T(G)\le 3^{n-2}$.
\end{theorem}
\begin{proof}
From each spanning tree of $G_{n-1}$, the algorithm generates
two spanning trees of $G_n$ in which $v_n$ is a leaf, and for each spanning
tree containing the edge joining the neighbors of $v_n$, it generates one
more tree in which $v_n$ is not a leaf.  Hence the bounds follow immediately
from Theorem~\ref{listall} by induction on $n$.
\end{proof}

The bounds in Theorem~\ref{lubound} are not sharp.  The constructions we
present in this section actually minimize and maximize $T(G)$ among $n$-vertex
$2$-trees $G$.  We will show in the next two sections that they are all the
extremal $2$-trees.  Recall that the $n$-book $B_n$ is the $n$-vertex $2$-tree
obtained from the complete graph $K_2$ with vertices $x$ and $y$ by iteratively
adding simplicial vertices with neighborhood $\{x,y\}$; it has $n-2$ simplicial
vertices.

\begin{theorem}\label{strictbound}
The $n$-vertex $2$-tree $B_n$ has $n2^{n-3}$ spanning trees.
\end{theorem}
\begin{proof}
To count the spanning trees directly, note that those containing the edge $xy$
simply partition the remaining vertices into neighbors of $x$ and neighbors of
$y$; hence there are $2^{n-2}$ such trees.  When $xy$ is not used, one of the
$n-2$ remaining vertices is adjacent to both $x$ and $y$, and the remaining
$n-3$ vertices are partitioned into
neighbors of $x$ and neighbors of $y$.  Hence
$T(B_n)=2^{n-2}+(n-2)2^{n-3}=n2^{n-3}$.
\end{proof}

It is well known that every chordal graph that is not a complete graph
has (at least) two nonadjacent simplicial vertices (Dirac \cite{dirac61}).

\begin{lemma}\label{2path}
An $n$-vertex $2$-tree $G$ has exactly two simplicial vertices if and only if
$G$ has a $2$-simplicial ordering $(v_n,\ldots,v_1)$ such that each
$v_i$ is adjacent to $v_{i-1}$, for $2\le i\le n$.  That is, the vertices
$v_n,\ldots,v_1$ form a path in order.
\end{lemma}
\begin{proof}
Since a vertex in a $2$-tree with $n\ge3$ is simplicial if and only if it has
degree $2$, the neighborhood in $G$ of a simplical vertex $v$ contains at most
one simplicial vertex of $G-v$ (if $n\ge3$).  Hence $G$ has at least as many
simplicial vertices as $G-v$.

The only $2$-tree with four vertices has exactly two simplicial vertices.  If
$n\ge5$, then $G-v$ has the same number of simplicial vertices as $G$ if and
only if $v$ is adjacent to some simplicial vertex of $G-v$.  Hence in
generating a larger $2$-tree from the $4$-vertex $2$-tree by the reverse of
a $2$-simplicial ordering, the number of simplicial vertices remains $2$ if and
only if each subsequent added vertex is adjacent to a current simplicial vertex.

When there are exactly two simplicial vertices, this process can be reversed,
because any simplicial vertex can be deleted next in a $2$-simplicial ordering.
In particular, we can {\it avoid} deleting one of the two original simplicial
vertices.  Since the number of simplicial vertices cannot decrease below $2$,
each vertex we delete from the ``other end'' is adjacent to one simplicial
vertex, thus forming the desired path.
\end{proof}

Various $2$-trees having exactly two simplicial vertices.  One is the
``square'' of the $n$-vertex path $P_n$, with vertices
$\{v_n,\ldots,v_1\}$ and edges $\{v_iv_j\colon\, |i-j|\le 2\}$.  Another is
obtained from $P_n$ by making one endpoint adjacent to all the other vertices.
The number of spanning trees in the latter graph has been well studied (for
example, see \cite{hilton}).  Surprisingly, all such $2$-trees with
$n$-vertices have the same number of spanning trees.  We prove a more detailed
statement for use later in proving the extremal result.

\begin{definition}\label{tseq}
For a set $S$ of edges in a graph $G$, let $T(G;S)$ denote the number of
spanning trees of $G$ containing $S$; we write simply $T(G;e)$ when $S=\{e\}$.
Also, the Fibonacci sequence $\la F\ra$ is defined by $F_0=0$,
$F_1=1$, and $F_n=F_{n-1}+F_{n-2}$ for  $n\ge2$.
(The solution formula is
$F_n=\FR1{\sqrt5}[\phi^n-(-\phi)^{-n}]$, where $\phi=(1+\sqrt5)/2$.)
\end{definition}


\begin{theorem}\label{maincount}
Let $e_0$ be an edge in a $2$-tree $\hH$, and let $\alpha=T(\hH)$ and
$\beta=T(\hH;e_0)$.  For $p\ge1$, let $G_p$ be a $2$-tree grown from $\hH$ by
successively adding simplicial vertices $\VEC w1p$ such that for $i\ge2$ the
neighbors of $w_i$ are the endpoints of an edge $e_{i-1}$ incident to $w_{i-1}$
($w_1$ is adjacent to the endpoints of $e_0$).  Letting $t_p=T(G_p)$ and
$s_p=T(G_p;e_p)$, we have
$$
t_p=F_{2p+1}\alpha+F_{2p}\beta
\qquad {\textrm and}\qquad s_p=F_{2p}\alpha+F_{2p-1}\beta.
$$
\end{theorem}
\begin{proof}
We use induction on $p$.  Set $G_0=\hH$, so $t_0=\alpha$ and $s_0=\beta$.
Extending the Fibonacci sequence by $F_{-1}=1$ (preserving the recurrence), the
claim holds for $p=0$.

Consider $p\ge1$.  As in Theorem~\ref{listall}, we have
$t_p=2t_{p-1}+s_{p-1}$ and $s_{p}=t_{p-1}+s_{p-1}$.  Using the induction
hypothesis, 
\begin{eqnarray*}
t_p&=&(2F_{2p-1}\alpha+2F_{2p-2}\beta)+(F_{2p-2}\alpha+F_{2p-3}\beta)
~=~ F_{2p+1}\alpha+F_{2p}\beta,\\
s_p&=&(F_{2p-1}\alpha+F_{2p-2}\beta)+(F_{2p-2}\alpha+F_{2p-3}\beta)
~=~F_{2p}\alpha+F_{2p-1}\beta.
\end{eqnarray*}

\vspace{-2.4pc}
\qquad
\end{proof}

\begin{corollary}\label{2simp}
For an $n$-vertex $2$-tree $G$ with exactly two simplicial vertices,
$T(G)=F_{2n-2}$.  The value is asymptotic to
$\FR1{\sqrt5}[(1+\sqrt5)/2]^{2n-2}$, which is approximately $.1708(2.618)^n$.
\end{corollary}
\begin{proof}
We grow $G$ from the $2$-tree with two vertices by $n-2$ steps that add
simplicial vertices.  Thus we apply Theorem~\ref{maincount} with $p=n-2$
and $\alpha=\beta=1$ to obtain $T(G)=F_{2n-2}$.
\end{proof}

\section{2-Trees with the Most Spanning Trees}

In this section we prove that the $n$-vertex $2$-trees with the most spanning
trees are those having exactly two simplicial vertices.  We do this by proving
that if $G$ is an $n$-vertex $2$-tree with more than two simplicial vertices,
then some $2$-tree with fewer simplicial vertices than $G$ has more spanning
trees than $G$.  We need additional computations based on
Theorem~\ref{maincount}.

\begin{lemma}\label{anyset}
Let $H$ and $J$ be $2$-trees properly containing the edge $xy$.
For every set $S\esub E(J)$ such that $S$ contains no cycle,
$T(H\cup J;S)$ is a strictly increasing function of $T(H;xy)$.
\end{lemma}

\begin{proof}
We use induction on $|V(J)|$; note $|V(J)|\ge3$.  Let $v$ be a simplicial
vertex of $J$ not in $\{x,y\}$, and let $J'=J-v$.  Let $N_J(v)=\{w,z\}$,
and let $e=wz$ (see Figure~\ref{figHJ}).  For $S\esub E(J')$, let
$t=T(H\cup J';S)$ and $s=T(H\cup J';S\cup\{e\})$.  Note that $S\cup\{e\}$ may
contain a cycle when $e\notin S$, in which case $s=0$.  Also $S\cup\{e\}=S$
when $e\in S$.  By considering which edges in a spanning tree of $H\cup J$ are 
incident to $v$ (as in earlier arguments), we compute

\medskip
\centerline{
\hbox{
\begin{tabular}{l|l|l}
 &if $e\notin S$&if $e\in S$\\
\hline
$T(H\cup J;S)$&$=2t+s$&$=2s$\\
$T(H\cup J;S\cup\{vw\})$&$=t+s$&$=s$\\
$T(H\cup J;S\cup\{vz\})$&$=t+s$&$=s$\\
$T(H\cup J;S\cup\{vw,vz\})$&$=s$&$=0$
\end{tabular}
}
}

\medskip
\noindent
When $|V(J)|=3$, the value $t$ is independent of $s$ and hence constant as a
function of $s$, but adding $s$ gives the desired behavior.
When $|V(J)|>3$ and the specified set is acyclic, by the induction hypothesis
all summands are strictly increasing in $s$.
\end{proof}

\vspace{-1pc}
\begin{figure}[h]
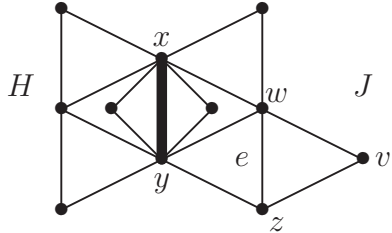

\gpic{
\expandafter\ifx\csname graph\endcsname\relax \csname newbox\endcsname\graph\fi
\expandafter\ifx\csname graphtemp\endcsname\relax \csname newdimen\endcsname\graphtemp\fi
\setbox\graph=\vtop{\vskip 0pt\hbox{%
    \graphtemp=.5ex\advance\graphtemp by 1.118in
    \rlap{\kern 0.316in\lower\graphtemp\hbox to 0pt{\hss $\bu$\hss}}%
    \graphtemp=.5ex\advance\graphtemp by 0.592in
    \rlap{\kern 0.316in\lower\graphtemp\hbox to 0pt{\hss $\bu$\hss}}%
    \graphtemp=.5ex\advance\graphtemp by 0.066in
    \rlap{\kern 0.316in\lower\graphtemp\hbox to 0pt{\hss $\bu$\hss}}%
    \graphtemp=.5ex\advance\graphtemp by 0.592in
    \rlap{\kern 0.579in\lower\graphtemp\hbox to 0pt{\hss $\bu$\hss}}%
    \graphtemp=.5ex\advance\graphtemp by 0.855in
    \rlap{\kern 0.842in\lower\graphtemp\hbox to 0pt{\hss $\bu$\hss}}%
    \graphtemp=.5ex\advance\graphtemp by 0.329in
    \rlap{\kern 0.842in\lower\graphtemp\hbox to 0pt{\hss $\bu$\hss}}%
    \graphtemp=.5ex\advance\graphtemp by 0.592in
    \rlap{\kern 1.105in\lower\graphtemp\hbox to 0pt{\hss $\bu$\hss}}%
    \graphtemp=.5ex\advance\graphtemp by 0.855in
    \rlap{\kern 1.895in\lower\graphtemp\hbox to 0pt{\hss $\bu$\hss}}%
    \special{pn 11}%
    \special{pa 316 1118}%
    \special{pa 316 592}%
    \special{fp}%
    \special{pa 316 592}%
    \special{pa 316 66}%
    \special{fp}%
    \special{pa 316 66}%
    \special{pa 842 329}%
    \special{fp}%
    \special{pa 842 329}%
    \special{pa 316 592}%
    \special{fp}%
    \special{pa 316 592}%
    \special{pa 842 855}%
    \special{fp}%
    \special{pa 842 855}%
    \special{pa 316 1118}%
    \special{fp}%
    \graphtemp=.5ex\advance\graphtemp by 1.118in
    \rlap{\kern 1.368in\lower\graphtemp\hbox to 0pt{\hss $\bu$\hss}}%
    \graphtemp=.5ex\advance\graphtemp by 0.592in
    \rlap{\kern 1.368in\lower\graphtemp\hbox to 0pt{\hss $\bu$\hss}}%
    \graphtemp=.5ex\advance\graphtemp by 0.066in
    \rlap{\kern 1.368in\lower\graphtemp\hbox to 0pt{\hss $\bu$\hss}}%
    \special{pa 579 592}%
    \special{pa 842 855}%
    \special{fp}%
    \special{pa 842 855}%
    \special{pa 1105 592}%
    \special{fp}%
    \special{pa 1105 592}%
    \special{pa 842 329}%
    \special{fp}%
    \special{pa 842 329}%
    \special{pa 579 592}%
    \special{fp}%
    \special{pn 56}%
    \special{pa 842 855}%
    \special{pa 842 329}%
    \special{fp}%
    \special{pn 11}%
    \special{pa 1368 1118}%
    \special{pa 842 855}%
    \special{fp}%
    \special{pa 842 855}%
    \special{pa 1368 592}%
    \special{fp}%
    \special{pa 1368 592}%
    \special{pa 1895 855}%
    \special{fp}%
    \special{pa 1895 855}%
    \special{pa 1368 1118}%
    \special{fp}%
    \special{pa 1368 1118}%
    \special{pa 1368 592}%
    \special{fp}%
    \special{pa 1368 592}%
    \special{pa 1368 66}%
    \special{fp}%
    \special{pa 1368 66}%
    \special{pa 842 329}%
    \special{fp}%
    \special{pa 842 329}%
    \special{pa 1368 592}%
    \special{fp}%
    \graphtemp=.5ex\advance\graphtemp by 0.961in
    \rlap{\kern 0.842in\lower\graphtemp\hbox to 0pt{\hss $y$\hss}}%
    \graphtemp=.5ex\advance\graphtemp by 0.224in
    \rlap{\kern 0.842in\lower\graphtemp\hbox to 0pt{\hss $x$\hss}}%
    \graphtemp=.5ex\advance\graphtemp by 0.855in
    \rlap{\kern 2.000in\lower\graphtemp\hbox to 0pt{\hss $v$\hss}}%
    \graphtemp=.5ex\advance\graphtemp by 1.193in
    \rlap{\kern 1.443in\lower\graphtemp\hbox to 0pt{\hss $z$\hss}}%
    \graphtemp=.5ex\advance\graphtemp by 0.518in
    \rlap{\kern 1.443in\lower\graphtemp\hbox to 0pt{\hss $w$\hss}}%
    \graphtemp=.5ex\advance\graphtemp by 0.855in
    \rlap{\kern 1.263in\lower\graphtemp\hbox to 0pt{\hss $e$\hss}}%
    \graphtemp=.5ex\advance\graphtemp by 0.487in
    \rlap{\kern 0.105in\lower\graphtemp\hbox to 0pt{\hss $H$\hss}}%
    \graphtemp=.5ex\advance\graphtemp by 0.487in
    \rlap{\kern 1.895in\lower\graphtemp\hbox to 0pt{\hss $J$\hss}}%
    \hbox{\vrule depth1.224in width0pt height 0pt}%
    \kern 2.000in
  }%
}%
}

\vspace{-1pc} 
\caption{$2$-trees $H$ and $J$ sharing an edge $xy$ in
Lemma~\ref{anyset}.\label{figHJ}}
\end{figure}

\begin{lemma}\label{pathineq}
Let $e_0$ be an edge in a $2$-tree $\hH$ with at least three vertices, and let
$\alpha=T(\hH)$ and $\beta=T(\hH;e_0)$.  For $p\ge1$, let $G_p$ be a $2$-tree
grown from $\hH$ by successively adding simplicial vertices $\VEC w1p$ such
that the neighbors of $w_i$ are the endpoints of an edge $e_{i-1}$ incident to
$w_{i-1}$ ($w_1$ is adjacent to the endpoints of $e_0$).  Let $e'$ be the edge
of $G_1$ other than $e_1$ that is incident to $w_1$ and $e_0$.  The numbers of 
spanning trees containing $e_0$ and $e'$ satisfy
$$
T(G_p;e_0)=F_{2p+1}\beta\qquad
{\textrm and}\qquad T(G_p;e')=F_{2p-1}\alpha+F_{2p}\beta
$$
In addition, $T(G_p;e_p)>T(G_p;e_0)$ and $T(G_p;e_p)>T(G_p;e')$.
\end{lemma}
\begin{proof}
See Figure~\ref{figGp}.
By Theorem~\ref{maincount}, $T(G_p;e_p)=F_{2p+1}\alpha+F_{2p}\beta$.
Hence the final statement follows immediately from the claimed formulas
for $T(G_p;e_0)$ and $T(G_p;e')$.

For $T(G_p;e_0)$, note that the vertices of $e_0$ form a separating $2$-set in
$G_p$.  Let $G'$ be the
subgraph of $G_p$ induced by the endpoints of $e_0$ and the vertices
$\VEC w1p$.  Every spanning tree in $G_p$ containing $e_0$ is the union of two
trees containing $e_0$: a spanning tree in $\hH$ and a spanning tree in $G'$. 
Note that $G'$ is a $2$-tree with two simplicial vertices and can be grown in
reverse order by switching the roles of $e_0$ and $e_p$.  Thus $T(G';e_0)$ is
the value of $s_p$ from Theorem~\ref{maincount} when starting with
$\alpha=\beta=1$, since $G'$ is grown from a $2$-tree consisting of one edge.
Multiplying by $T(G_0;e_0)$ yields
$T(G_p;e_0)=(F_{2p}+F_{2p-1})\beta=F_{2p+1}\beta$.

For $T(G_p;e')$, we subtract from $T(G_p)$ the number $\tau$ of spanning trees
of $G_p$ that do not contain $e'$.  Deleting $e'$ leaves the common vertex of
$e_0$ and $e_1$ as a cut-vertex.  Hence $\tau$ is the product of the numbers of
spanning trees in the two blocks of $G_p-e'$.  One block is just $G_0$.  The
other is a $2$-tree with $p+1$ vertices having exactly two simplicial vertices.
By Corollary~\ref{2simp}, it has $F_{2p}$ spanning trees.
By Theorem~\ref{maincount}, $T(G_p)=F_{2p+1}\alpha+F_{2p}\beta$.
Altogether, we have
$T(G_p;e')=F_{2p+1}\alpha+F_{2p}\beta-F_{2p}\alpha=F_{2p-1}\alpha+F_{2p}\beta.$
\end{proof}

\vspace{-1pc}
\begin{figure}[h]
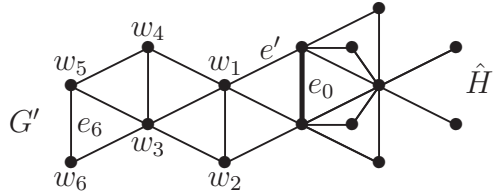

\gpic{
\expandafter\ifx\csname graph\endcsname\relax \csname newbox\endcsname\graph\fi
\expandafter\ifx\csname graphtemp\endcsname\relax \csname newdimen\endcsname\graphtemp\fi
\setbox\graph=\vtop{\vskip 0pt\hbox{%
    \graphtemp=.5ex\advance\graphtemp by 0.856in
    \rlap{\kern 0.403in\lower\graphtemp\hbox to 0pt{\hss $\bu$\hss}}%
    \graphtemp=.5ex\advance\graphtemp by 0.453in
    \rlap{\kern 0.403in\lower\graphtemp\hbox to 0pt{\hss $\bu$\hss}}%
    \graphtemp=.5ex\advance\graphtemp by 0.655in
    \rlap{\kern 0.806in\lower\graphtemp\hbox to 0pt{\hss $\bu$\hss}}%
    \graphtemp=.5ex\advance\graphtemp by 0.252in
    \rlap{\kern 0.806in\lower\graphtemp\hbox to 0pt{\hss $\bu$\hss}}%
    \graphtemp=.5ex\advance\graphtemp by 0.856in
    \rlap{\kern 1.209in\lower\graphtemp\hbox to 0pt{\hss $\bu$\hss}}%
    \graphtemp=.5ex\advance\graphtemp by 0.453in
    \rlap{\kern 1.209in\lower\graphtemp\hbox to 0pt{\hss $\bu$\hss}}%
    \graphtemp=.5ex\advance\graphtemp by 0.655in
    \rlap{\kern 1.612in\lower\graphtemp\hbox to 0pt{\hss $\bu$\hss}}%
    \graphtemp=.5ex\advance\graphtemp by 0.252in
    \rlap{\kern 1.612in\lower\graphtemp\hbox to 0pt{\hss $\bu$\hss}}%
    \graphtemp=.5ex\advance\graphtemp by 0.252in
    \rlap{\kern 2.418in\lower\graphtemp\hbox to 0pt{\hss $\bu$\hss}}%
    \special{pn 11}%
    \special{pa 403 856}%
    \special{pa 403 453}%
    \special{fp}%
    \special{pa 403 453}%
    \special{pa 806 252}%
    \special{fp}%
    \special{pa 806 252}%
    \special{pa 806 655}%
    \special{fp}%
    \special{pa 806 655}%
    \special{pa 403 856}%
    \special{fp}%
    \special{pa 403 453}%
    \special{pa 806 655}%
    \special{fp}%
    \special{pa 806 655}%
    \special{pa 1209 856}%
    \special{fp}%
    \special{pa 1209 856}%
    \special{pa 1612 655}%
    \special{fp}%
    \special{pa 1209 453}%
    \special{pa 806 252}%
    \special{fp}%
    \special{pa 1209 453}%
    \special{pa 806 655}%
    \special{fp}%
    \special{pa 1209 453}%
    \special{pa 1209 856}%
    \special{fp}%
    \special{pa 1209 453}%
    \special{pa 1612 655}%
    \special{fp}%
    \special{pa 1209 453}%
    \special{pa 1612 252}%
    \special{fp}%
    \special{pn 28}%
    \special{pa 1612 655}%
    \special{pa 1612 252}%
    \special{fp}%
    \graphtemp=.5ex\advance\graphtemp by 0.937in
    \rlap{\kern 0.403in\lower\graphtemp\hbox to 0pt{\hss $w_6$\hss}}%
    \graphtemp=.5ex\advance\graphtemp by 0.332in
    \rlap{\kern 0.403in\lower\graphtemp\hbox to 0pt{\hss $w_5$\hss}}%
    \graphtemp=.5ex\advance\graphtemp by 0.756in
    \rlap{\kern 0.806in\lower\graphtemp\hbox to 0pt{\hss $w_3$\hss}}%
    \graphtemp=.5ex\advance\graphtemp by 0.151in
    \rlap{\kern 0.806in\lower\graphtemp\hbox to 0pt{\hss $w_4$\hss}}%
    \graphtemp=.5ex\advance\graphtemp by 0.937in
    \rlap{\kern 1.209in\lower\graphtemp\hbox to 0pt{\hss $w_2$\hss}}%
    \graphtemp=.5ex\advance\graphtemp by 0.332in
    \rlap{\kern 1.209in\lower\graphtemp\hbox to 0pt{\hss $w_1$\hss}}%
    \graphtemp=.5ex\advance\graphtemp by 0.275in
    \rlap{\kern 1.447in\lower\graphtemp\hbox to 0pt{\hss $e'$\hss}}%
    \graphtemp=.5ex\advance\graphtemp by 0.453in
    \rlap{\kern 1.713in\lower\graphtemp\hbox to 0pt{\hss $e_0$\hss}}%
    \graphtemp=.5ex\advance\graphtemp by 0.655in
    \rlap{\kern 0.504in\lower\graphtemp\hbox to 0pt{\hss $e_6$\hss}}%
    \graphtemp=.5ex\advance\graphtemp by 0.453in
    \rlap{\kern 2.015in\lower\graphtemp\hbox to 0pt{\hss $\bu$\hss}}%
    \graphtemp=.5ex\advance\graphtemp by 0.050in
    \rlap{\kern 2.015in\lower\graphtemp\hbox to 0pt{\hss $\bu$\hss}}%
    \graphtemp=.5ex\advance\graphtemp by 0.252in
    \rlap{\kern 1.874in\lower\graphtemp\hbox to 0pt{\hss $\bu$\hss}}%
    \graphtemp=.5ex\advance\graphtemp by 0.655in
    \rlap{\kern 1.874in\lower\graphtemp\hbox to 0pt{\hss $\bu$\hss}}%
    \graphtemp=.5ex\advance\graphtemp by 0.856in
    \rlap{\kern 2.015in\lower\graphtemp\hbox to 0pt{\hss $\bu$\hss}}%
    \graphtemp=.5ex\advance\graphtemp by 0.655in
    \rlap{\kern 2.418in\lower\graphtemp\hbox to 0pt{\hss $\bu$\hss}}%
    \special{pn 11}%
    \special{pa 2015 453}%
    \special{pa 2015 50}%
    \special{fp}%
    \special{pa 2015 453}%
    \special{pa 1874 252}%
    \special{fp}%
    \special{pa 2015 453}%
    \special{pa 1612 252}%
    \special{fp}%
    \special{pa 2015 453}%
    \special{pa 1612 655}%
    \special{fp}%
    \special{pa 2015 453}%
    \special{pa 1874 655}%
    \special{fp}%
    \special{pa 2015 453}%
    \special{pa 2015 856}%
    \special{fp}%
    \special{pa 2015 453}%
    \special{pa 2418 655}%
    \special{fp}%
    \special{pa 2015 453}%
    \special{pa 2418 252}%
    \special{fp}%
    \special{pa 2015 50}%
    \special{pa 1612 252}%
    \special{fp}%
    \special{pa 1612 252}%
    \special{pa 1874 252}%
    \special{fp}%
    \special{pa 1874 252}%
    \special{pa 2015 453}%
    \special{fp}%
    \special{pa 2015 453}%
    \special{pa 1874 655}%
    \special{fp}%
    \special{pa 1874 655}%
    \special{pa 1612 655}%
    \special{fp}%
    \special{pa 1612 655}%
    \special{pa 2015 856}%
    \special{fp}%
    \special{pa 2015 856}%
    \special{pa 1612 655}%
    \special{fp}%
    \special{pa 1612 655}%
    \special{pa 2418 252}%
    \special{fp}%
    \graphtemp=.5ex\advance\graphtemp by 0.655in
    \rlap{\kern 0.161in\lower\graphtemp\hbox to 0pt{\hss $G'$\hss}}%
    \graphtemp=.5ex\advance\graphtemp by 0.453in
    \rlap{\kern 2.539in\lower\graphtemp\hbox to 0pt{\hss $\hH$\hss}}%
    \hbox{\vrule depth0.937in width0pt height 0pt}%
    \kern 2.700in
  }%
}%
}

\vspace{-1pc} 
\caption{Growing $G_6$ from $\hH$ containing $e_0$ in 
Lemma~\ref{pathineq}.\label{figGp}}
\end{figure}

For an $n$-vertex $2$-tree $G$ with more than two simplicial vertices, an
appropriate decomposition of $G$ will allow us to invoke these two lemmas to
obtain an $n$-vertex $2$-tree with more spanning trees.

\begin{theorem}\label{mainthm}
If $n\ge4$ and $G$ is an $n$-vertex $2$-tree with more than two simplicial
vertices, then $G$ has fewer spanning trees than the $n$-vertex $2$-trees with
two simplicial vertices.
\end{theorem}
\begin{proof}
There is only one $4$-vertex $2$-tree, and it has two simplicial vertices.
For $G$ as specified, we obtain an $n$-vertex $2$-tree $G'$ with more spanning
trees.  We do this by expressing $G$ as $H\cup J$ with one common edge to apply
Lemma~\ref{anyset}, and we will apply Lemma~\ref{pathineq} after forming $G'$
by attaching $J$ at an edge belonging to more spanning trees than $e$ in $H$.

Let $v$ and $v'$ be two simplicial vertices in $G$; they are not adjacent.
While there remain other simplicial vertices, iteratively delete a simplicial
vertex not in $\{v,v'\}$.  This begins a $2$-simplicial ordering of $G$.  When
the process cannot continue further, what remains is a $2$-tree $H'$ contained
in $G$; the only simplicial vertices of $H'$ are $v$ and $v'$.  See
Figure~\ref{figmain}.

\eject
\begin{figure}[h!]
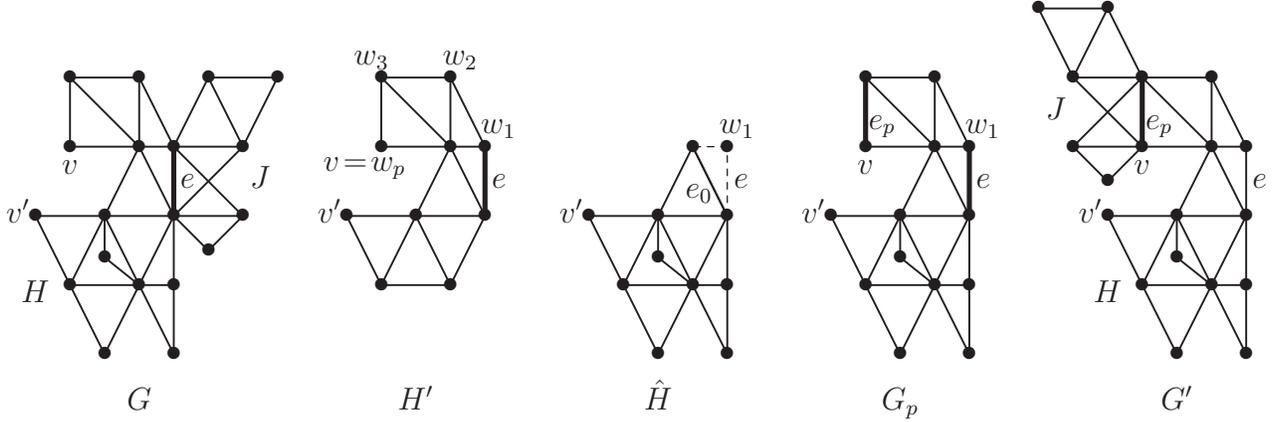

\gpic{
\expandafter\ifx\csname graph\endcsname\relax \csname newbox\endcsname\graph\fi
\expandafter\ifx\csname graphtemp\endcsname\relax \csname newdimen\endcsname\graphtemp\fi
\setbox\graph=\vtop{\vskip 0pt\hbox{%
    \graphtemp=.5ex\advance\graphtemp by 1.132in
    \rlap{\kern 0.091in\lower\graphtemp\hbox to 0pt{\hss $\bu$\hss}}%
    \graphtemp=.5ex\advance\graphtemp by 1.494in
    \rlap{\kern 0.272in\lower\graphtemp\hbox to 0pt{\hss $\bu$\hss}}%
    \graphtemp=.5ex\advance\graphtemp by 1.132in
    \rlap{\kern 0.453in\lower\graphtemp\hbox to 0pt{\hss $\bu$\hss}}%
    \graphtemp=.5ex\advance\graphtemp by 1.494in
    \rlap{\kern 0.634in\lower\graphtemp\hbox to 0pt{\hss $\bu$\hss}}%
    \graphtemp=.5ex\advance\graphtemp by 1.132in
    \rlap{\kern 0.815in\lower\graphtemp\hbox to 0pt{\hss $\bu$\hss}}%
    \graphtemp=.5ex\advance\graphtemp by 0.769in
    \rlap{\kern 0.815in\lower\graphtemp\hbox to 0pt{\hss $\bu$\hss}}%
    \graphtemp=.5ex\advance\graphtemp by 0.769in
    \rlap{\kern 0.634in\lower\graphtemp\hbox to 0pt{\hss $\bu$\hss}}%
    \graphtemp=.5ex\advance\graphtemp by 0.407in
    \rlap{\kern 0.634in\lower\graphtemp\hbox to 0pt{\hss $\bu$\hss}}%
    \graphtemp=.5ex\advance\graphtemp by 0.407in
    \rlap{\kern 0.272in\lower\graphtemp\hbox to 0pt{\hss $\bu$\hss}}%
    \graphtemp=.5ex\advance\graphtemp by 0.769in
    \rlap{\kern 0.272in\lower\graphtemp\hbox to 0pt{\hss $\bu$\hss}}%
    \special{pn 11}%
    \special{pa 91 1132}%
    \special{pa 272 1494}%
    \special{fp}%
    \special{pa 272 1494}%
    \special{pa 634 1494}%
    \special{fp}%
    \special{pa 634 1494}%
    \special{pa 815 1132}%
    \special{fp}%
    \special{pa 815 1132}%
    \special{pa 815 769}%
    \special{fp}%
    \special{pa 815 769}%
    \special{pa 634 407}%
    \special{fp}%
    \special{pa 634 407}%
    \special{pa 272 407}%
    \special{fp}%
    \special{pa 272 407}%
    \special{pa 634 769}%
    \special{fp}%
    \special{pa 634 769}%
    \special{pa 453 1132}%
    \special{fp}%
    \special{pa 453 1132}%
    \special{pa 91 1132}%
    \special{fp}%
    \special{pa 453 1132}%
    \special{pa 272 1494}%
    \special{fp}%
    \special{pa 453 1132}%
    \special{pa 634 1494}%
    \special{fp}%
    \special{pa 453 1132}%
    \special{pa 815 1132}%
    \special{fp}%
    \special{pa 272 407}%
    \special{pa 272 769}%
    \special{fp}%
    \special{pa 272 769}%
    \special{pa 634 769}%
    \special{fp}%
    \special{pa 634 769}%
    \special{pa 815 769}%
    \special{fp}%
    \special{pa 634 407}%
    \special{pa 634 769}%
    \special{fp}%
    \special{pa 634 769}%
    \special{pa 815 1132}%
    \special{fp}%
    \graphtemp=.5ex\advance\graphtemp by 1.132in
    \rlap{\kern 0.000in\lower\graphtemp\hbox to 0pt{\hss $v'$\hss}}%
    \graphtemp=.5ex\advance\graphtemp by 0.878in
    \rlap{\kern 0.272in\lower\graphtemp\hbox to 0pt{\hss $v$\hss}}%
    \graphtemp=.5ex\advance\graphtemp by 0.951in
    \rlap{\kern 0.887in\lower\graphtemp\hbox to 0pt{\hss $e$\hss}}%
    \special{pn 28}%
    \special{pa 815 1132}%
    \special{pa 815 769}%
    \special{fp}%
    \graphtemp=.5ex\advance\graphtemp by 1.313in
    \rlap{\kern 0.996in\lower\graphtemp\hbox to 0pt{\hss $\bu$\hss}}%
    \graphtemp=.5ex\advance\graphtemp by 1.132in
    \rlap{\kern 1.177in\lower\graphtemp\hbox to 0pt{\hss $\bu$\hss}}%
    \graphtemp=.5ex\advance\graphtemp by 0.769in
    \rlap{\kern 1.177in\lower\graphtemp\hbox to 0pt{\hss $\bu$\hss}}%
    \graphtemp=.5ex\advance\graphtemp by 0.407in
    \rlap{\kern 0.996in\lower\graphtemp\hbox to 0pt{\hss $\bu$\hss}}%
    \graphtemp=.5ex\advance\graphtemp by 0.407in
    \rlap{\kern 1.358in\lower\graphtemp\hbox to 0pt{\hss $\bu$\hss}}%
    \special{pn 11}%
    \special{pa 815 1132}%
    \special{pa 996 1313}%
    \special{fp}%
    \special{pa 996 1313}%
    \special{pa 1177 1132}%
    \special{fp}%
    \special{pa 1177 1132}%
    \special{pa 815 1132}%
    \special{fp}%
    \special{pa 815 1132}%
    \special{pa 1177 769}%
    \special{fp}%
    \special{pa 1177 769}%
    \special{pa 1358 407}%
    \special{fp}%
    \special{pa 1358 407}%
    \special{pa 996 407}%
    \special{fp}%
    \special{pa 996 407}%
    \special{pa 815 769}%
    \special{fp}%
    \special{pa 815 769}%
    \special{pa 1177 769}%
    \special{fp}%
    \special{pa 1177 1132}%
    \special{pa 815 769}%
    \special{fp}%
    \special{pa 1177 769}%
    \special{pa 996 407}%
    \special{fp}%
    \graphtemp=.5ex\advance\graphtemp by 1.349in
    \rlap{\kern 0.453in\lower\graphtemp\hbox to 0pt{\hss $\bu$\hss}}%
    \graphtemp=.5ex\advance\graphtemp by 1.494in
    \rlap{\kern 0.815in\lower\graphtemp\hbox to 0pt{\hss $\bu$\hss}}%
    \graphtemp=.5ex\advance\graphtemp by 1.856in
    \rlap{\kern 0.815in\lower\graphtemp\hbox to 0pt{\hss $\bu$\hss}}%
    \graphtemp=.5ex\advance\graphtemp by 1.856in
    \rlap{\kern 0.453in\lower\graphtemp\hbox to 0pt{\hss $\bu$\hss}}%
    \special{pa 453 1132}%
    \special{pa 453 1349}%
    \special{fp}%
    \special{pa 453 1349}%
    \special{pa 634 1494}%
    \special{fp}%
    \special{pa 634 1494}%
    \special{pa 815 1494}%
    \special{fp}%
    \special{pa 815 1494}%
    \special{pa 815 1856}%
    \special{fp}%
    \special{pa 815 1856}%
    \special{pa 634 1494}%
    \special{fp}%
    \special{pa 634 1494}%
    \special{pa 453 1856}%
    \special{fp}%
    \special{pa 453 1856}%
    \special{pa 272 1494}%
    \special{fp}%
    \special{pa 815 1494}%
    \special{pa 815 1132}%
    \special{fp}%
    \graphtemp=.5ex\advance\graphtemp by 1.548in
    \rlap{\kern 0.091in\lower\graphtemp\hbox to 0pt{\hss $H$\hss}}%
    \graphtemp=.5ex\advance\graphtemp by 0.951in
    \rlap{\kern 1.267in\lower\graphtemp\hbox to 0pt{\hss $J$\hss}}%
    \graphtemp=.5ex\advance\graphtemp by 2.109in
    \rlap{\kern 0.634in\lower\graphtemp\hbox to 0pt{\hss $G$\hss}}%
    \graphtemp=.5ex\advance\graphtemp by 1.132in
    \rlap{\kern 1.720in\lower\graphtemp\hbox to 0pt{\hss $\bu$\hss}}%
    \graphtemp=.5ex\advance\graphtemp by 1.494in
    \rlap{\kern 1.901in\lower\graphtemp\hbox to 0pt{\hss $\bu$\hss}}%
    \graphtemp=.5ex\advance\graphtemp by 1.132in
    \rlap{\kern 2.082in\lower\graphtemp\hbox to 0pt{\hss $\bu$\hss}}%
    \graphtemp=.5ex\advance\graphtemp by 1.494in
    \rlap{\kern 2.263in\lower\graphtemp\hbox to 0pt{\hss $\bu$\hss}}%
    \graphtemp=.5ex\advance\graphtemp by 1.132in
    \rlap{\kern 2.444in\lower\graphtemp\hbox to 0pt{\hss $\bu$\hss}}%
    \graphtemp=.5ex\advance\graphtemp by 0.769in
    \rlap{\kern 2.444in\lower\graphtemp\hbox to 0pt{\hss $\bu$\hss}}%
    \graphtemp=.5ex\advance\graphtemp by 0.769in
    \rlap{\kern 2.263in\lower\graphtemp\hbox to 0pt{\hss $\bu$\hss}}%
    \graphtemp=.5ex\advance\graphtemp by 0.407in
    \rlap{\kern 2.263in\lower\graphtemp\hbox to 0pt{\hss $\bu$\hss}}%
    \graphtemp=.5ex\advance\graphtemp by 0.407in
    \rlap{\kern 1.901in\lower\graphtemp\hbox to 0pt{\hss $\bu$\hss}}%
    \graphtemp=.5ex\advance\graphtemp by 0.769in
    \rlap{\kern 1.901in\lower\graphtemp\hbox to 0pt{\hss $\bu$\hss}}%
    \graphtemp=.5ex\advance\graphtemp by 0.302in
    \rlap{\kern 1.850in\lower\graphtemp\hbox to 0pt{\hss $w_3$\hss}}%
    \graphtemp=.5ex\advance\graphtemp by 0.302in
    \rlap{\kern 2.314in\lower\graphtemp\hbox to 0pt{\hss $w_2$\hss}}%
    \graphtemp=.5ex\advance\graphtemp by 0.664in
    \rlap{\kern 2.514in\lower\graphtemp\hbox to 0pt{\hss $w_1$\hss}}%
    \special{pa 1720 1132}%
    \special{pa 1901 1494}%
    \special{fp}%
    \special{pa 1901 1494}%
    \special{pa 2263 1494}%
    \special{fp}%
    \special{pa 2263 1494}%
    \special{pa 2444 1132}%
    \special{fp}%
    \special{pa 2444 1132}%
    \special{pa 2444 769}%
    \special{fp}%
    \special{pa 2444 769}%
    \special{pa 2263 407}%
    \special{fp}%
    \special{pa 2263 407}%
    \special{pa 1901 407}%
    \special{fp}%
    \special{pa 1901 407}%
    \special{pa 2263 769}%
    \special{fp}%
    \special{pa 2263 769}%
    \special{pa 2082 1132}%
    \special{fp}%
    \special{pa 2082 1132}%
    \special{pa 1720 1132}%
    \special{fp}%
    \special{pa 2082 1132}%
    \special{pa 1901 1494}%
    \special{fp}%
    \special{pa 2082 1132}%
    \special{pa 2263 1494}%
    \special{fp}%
    \special{pa 2082 1132}%
    \special{pa 2444 1132}%
    \special{fp}%
    \special{pa 1901 407}%
    \special{pa 1901 769}%
    \special{fp}%
    \special{pa 1901 769}%
    \special{pa 2263 769}%
    \special{fp}%
    \special{pa 2263 769}%
    \special{pa 2444 769}%
    \special{fp}%
    \special{pa 2263 407}%
    \special{pa 2263 769}%
    \special{fp}%
    \special{pa 2263 769}%
    \special{pa 2444 1132}%
    \special{fp}%
    \graphtemp=.5ex\advance\graphtemp by 1.132in
    \rlap{\kern 1.630in\lower\graphtemp\hbox to 0pt{\hss $v'$\hss}}%
    \graphtemp=.5ex\advance\graphtemp by 0.857in
    \rlap{\kern 1.814in\lower\graphtemp\hbox to 0pt{\hss $v\!=\!w_p$\hss}}%
    \graphtemp=.5ex\advance\graphtemp by 0.951in
    \rlap{\kern 2.517in\lower\graphtemp\hbox to 0pt{\hss $e$\hss}}%
    \special{pn 28}%
    \special{pa 2444 1132}%
    \special{pa 2444 769}%
    \special{fp}%
    \graphtemp=.5ex\advance\graphtemp by 2.109in
    \rlap{\kern 2.082in\lower\graphtemp\hbox to 0pt{\hss $H'$\hss}}%
    \graphtemp=.5ex\advance\graphtemp by 1.132in
    \rlap{\kern 2.987in\lower\graphtemp\hbox to 0pt{\hss $\bu$\hss}}%
    \graphtemp=.5ex\advance\graphtemp by 1.494in
    \rlap{\kern 3.169in\lower\graphtemp\hbox to 0pt{\hss $\bu$\hss}}%
    \graphtemp=.5ex\advance\graphtemp by 1.132in
    \rlap{\kern 3.350in\lower\graphtemp\hbox to 0pt{\hss $\bu$\hss}}%
    \graphtemp=.5ex\advance\graphtemp by 1.494in
    \rlap{\kern 3.531in\lower\graphtemp\hbox to 0pt{\hss $\bu$\hss}}%
    \graphtemp=.5ex\advance\graphtemp by 1.132in
    \rlap{\kern 3.712in\lower\graphtemp\hbox to 0pt{\hss $\bu$\hss}}%
    \graphtemp=.5ex\advance\graphtemp by 0.769in
    \rlap{\kern 3.712in\lower\graphtemp\hbox to 0pt{\hss $\bu$\hss}}%
    \graphtemp=.5ex\advance\graphtemp by 0.769in
    \rlap{\kern 3.531in\lower\graphtemp\hbox to 0pt{\hss $\bu$\hss}}%
    \special{pn 11}%
    \special{pa 2987 1132}%
    \special{pa 3169 1494}%
    \special{fp}%
    \special{pa 3169 1494}%
    \special{pa 3531 1494}%
    \special{fp}%
    \special{pa 3531 1494}%
    \special{pa 3712 1132}%
    \special{fp}%
    \special{pa 3712 1132}%
    \special{pa 3531 769}%
    \special{fp}%
    \special{pa 3531 769}%
    \special{pa 3350 1132}%
    \special{fp}%
    \special{pa 3350 1132}%
    \special{pa 2987 1132}%
    \special{fp}%
    \special{pa 3350 1132}%
    \special{pa 3169 1494}%
    \special{fp}%
    \special{pa 3350 1132}%
    \special{pa 3531 1494}%
    \special{fp}%
    \special{pa 3350 1132}%
    \special{pa 3712 1132}%
    \special{fp}%
    \special{pa 3531 769}%
    \special{pa 3712 1132}%
    \special{fp}%
    \special{pn 8}%
    \special{pa 3712 1132}%
    \special{pa 3712 769}%
    \special{pa 3531 769}%
    \special{da 0.036}%
    \graphtemp=.5ex\advance\graphtemp by 1.132in
    \rlap{\kern 2.897in\lower\graphtemp\hbox to 0pt{\hss $v'$\hss}}%
    \graphtemp=.5ex\advance\graphtemp by 0.878in
    \rlap{\kern 1.901in\lower\graphtemp\hbox to 0pt{\hss $~$\hss}}%
    \graphtemp=.5ex\advance\graphtemp by 0.951in
    \rlap{\kern 3.784in\lower\graphtemp\hbox to 0pt{\hss $e$\hss}}%
    \graphtemp=.5ex\advance\graphtemp by 1.349in
    \rlap{\kern 3.350in\lower\graphtemp\hbox to 0pt{\hss $\bu$\hss}}%
    \graphtemp=.5ex\advance\graphtemp by 1.494in
    \rlap{\kern 3.712in\lower\graphtemp\hbox to 0pt{\hss $\bu$\hss}}%
    \graphtemp=.5ex\advance\graphtemp by 1.856in
    \rlap{\kern 3.712in\lower\graphtemp\hbox to 0pt{\hss $\bu$\hss}}%
    \graphtemp=.5ex\advance\graphtemp by 1.856in
    \rlap{\kern 3.350in\lower\graphtemp\hbox to 0pt{\hss $\bu$\hss}}%
    \special{pn 11}%
    \special{pa 3350 1132}%
    \special{pa 3350 1349}%
    \special{fp}%
    \special{pa 3350 1349}%
    \special{pa 3531 1494}%
    \special{fp}%
    \special{pa 3531 1494}%
    \special{pa 3712 1494}%
    \special{fp}%
    \special{pa 3712 1494}%
    \special{pa 3712 1856}%
    \special{fp}%
    \special{pa 3712 1856}%
    \special{pa 3531 1494}%
    \special{fp}%
    \special{pa 3531 1494}%
    \special{pa 3350 1856}%
    \special{fp}%
    \special{pa 3350 1856}%
    \special{pa 3169 1494}%
    \special{fp}%
    \special{pa 3712 1494}%
    \special{pa 3712 1132}%
    \special{fp}%
    \graphtemp=.5ex\advance\graphtemp by 2.109in
    \rlap{\kern 3.350in\lower\graphtemp\hbox to 0pt{\hss $\hH$\hss}}%
    \graphtemp=.5ex\advance\graphtemp by 0.664in
    \rlap{\kern 3.763in\lower\graphtemp\hbox to 0pt{\hss $w_1$\hss}}%
    \graphtemp=.5ex\advance\graphtemp by 1.002in
    \rlap{\kern 3.564in\lower\graphtemp\hbox to 0pt{\hss $e_0$\hss}}%
    \graphtemp=.5ex\advance\graphtemp by 1.132in
    \rlap{\kern 4.255in\lower\graphtemp\hbox to 0pt{\hss $\bu$\hss}}%
    \graphtemp=.5ex\advance\graphtemp by 1.494in
    \rlap{\kern 4.436in\lower\graphtemp\hbox to 0pt{\hss $\bu$\hss}}%
    \graphtemp=.5ex\advance\graphtemp by 1.132in
    \rlap{\kern 4.617in\lower\graphtemp\hbox to 0pt{\hss $\bu$\hss}}%
    \graphtemp=.5ex\advance\graphtemp by 1.494in
    \rlap{\kern 4.798in\lower\graphtemp\hbox to 0pt{\hss $\bu$\hss}}%
    \graphtemp=.5ex\advance\graphtemp by 1.132in
    \rlap{\kern 4.979in\lower\graphtemp\hbox to 0pt{\hss $\bu$\hss}}%
    \graphtemp=.5ex\advance\graphtemp by 0.769in
    \rlap{\kern 4.979in\lower\graphtemp\hbox to 0pt{\hss $\bu$\hss}}%
    \graphtemp=.5ex\advance\graphtemp by 0.769in
    \rlap{\kern 4.798in\lower\graphtemp\hbox to 0pt{\hss $\bu$\hss}}%
    \graphtemp=.5ex\advance\graphtemp by 0.407in
    \rlap{\kern 4.798in\lower\graphtemp\hbox to 0pt{\hss $\bu$\hss}}%
    \graphtemp=.5ex\advance\graphtemp by 0.407in
    \rlap{\kern 4.436in\lower\graphtemp\hbox to 0pt{\hss $\bu$\hss}}%
    \graphtemp=.5ex\advance\graphtemp by 0.769in
    \rlap{\kern 4.436in\lower\graphtemp\hbox to 0pt{\hss $\bu$\hss}}%
    \special{pa 4255 1132}%
    \special{pa 4436 1494}%
    \special{fp}%
    \special{pa 4436 1494}%
    \special{pa 4798 1494}%
    \special{fp}%
    \special{pa 4798 1494}%
    \special{pa 4979 1132}%
    \special{fp}%
    \special{pa 4979 1132}%
    \special{pa 4979 769}%
    \special{fp}%
    \special{pa 4979 769}%
    \special{pa 4798 407}%
    \special{fp}%
    \special{pa 4798 407}%
    \special{pa 4436 407}%
    \special{fp}%
    \special{pa 4436 407}%
    \special{pa 4798 769}%
    \special{fp}%
    \special{pa 4798 769}%
    \special{pa 4617 1132}%
    \special{fp}%
    \special{pa 4617 1132}%
    \special{pa 4255 1132}%
    \special{fp}%
    \special{pa 4617 1132}%
    \special{pa 4436 1494}%
    \special{fp}%
    \special{pa 4617 1132}%
    \special{pa 4798 1494}%
    \special{fp}%
    \special{pa 4617 1132}%
    \special{pa 4979 1132}%
    \special{fp}%
    \special{pa 4436 407}%
    \special{pa 4436 769}%
    \special{fp}%
    \special{pa 4436 769}%
    \special{pa 4798 769}%
    \special{fp}%
    \special{pa 4798 769}%
    \special{pa 4979 769}%
    \special{fp}%
    \special{pa 4798 407}%
    \special{pa 4798 769}%
    \special{fp}%
    \special{pa 4798 769}%
    \special{pa 4979 1132}%
    \special{fp}%
    \graphtemp=.5ex\advance\graphtemp by 1.132in
    \rlap{\kern 4.164in\lower\graphtemp\hbox to 0pt{\hss $v'$\hss}}%
    \graphtemp=.5ex\advance\graphtemp by 0.878in
    \rlap{\kern 4.436in\lower\graphtemp\hbox to 0pt{\hss $v$\hss}}%
    \graphtemp=.5ex\advance\graphtemp by 0.951in
    \rlap{\kern 5.052in\lower\graphtemp\hbox to 0pt{\hss $e$\hss}}%
    \special{pn 28}%
    \special{pa 4979 1132}%
    \special{pa 4979 769}%
    \special{fp}%
    \special{pa 4436 769}%
    \special{pa 4436 407}%
    \special{fp}%
    \graphtemp=.5ex\advance\graphtemp by 1.349in
    \rlap{\kern 4.617in\lower\graphtemp\hbox to 0pt{\hss $\bu$\hss}}%
    \graphtemp=.5ex\advance\graphtemp by 1.494in
    \rlap{\kern 4.979in\lower\graphtemp\hbox to 0pt{\hss $\bu$\hss}}%
    \graphtemp=.5ex\advance\graphtemp by 1.856in
    \rlap{\kern 4.979in\lower\graphtemp\hbox to 0pt{\hss $\bu$\hss}}%
    \graphtemp=.5ex\advance\graphtemp by 1.856in
    \rlap{\kern 4.617in\lower\graphtemp\hbox to 0pt{\hss $\bu$\hss}}%
    \special{pn 11}%
    \special{pa 4617 1132}%
    \special{pa 4617 1349}%
    \special{fp}%
    \special{pa 4617 1349}%
    \special{pa 4798 1494}%
    \special{fp}%
    \special{pa 4798 1494}%
    \special{pa 4979 1494}%
    \special{fp}%
    \special{pa 4979 1494}%
    \special{pa 4979 1856}%
    \special{fp}%
    \special{pa 4979 1856}%
    \special{pa 4798 1494}%
    \special{fp}%
    \special{pa 4798 1494}%
    \special{pa 4617 1856}%
    \special{fp}%
    \special{pa 4617 1856}%
    \special{pa 4436 1494}%
    \special{fp}%
    \special{pa 4979 1494}%
    \special{pa 4979 1132}%
    \special{fp}%
    \graphtemp=.5ex\advance\graphtemp by 2.109in
    \rlap{\kern 4.617in\lower\graphtemp\hbox to 0pt{\hss $G_p$\hss}}%
    \graphtemp=.5ex\advance\graphtemp by 0.640in
    \rlap{\kern 4.523in\lower\graphtemp\hbox to 0pt{\hss $e_p$\hss}}%
    \graphtemp=.5ex\advance\graphtemp by 0.664in
    \rlap{\kern 5.048in\lower\graphtemp\hbox to 0pt{\hss $w_1$\hss}}%
    \graphtemp=.5ex\advance\graphtemp by 1.132in
    \rlap{\kern 5.703in\lower\graphtemp\hbox to 0pt{\hss $\bu$\hss}}%
    \graphtemp=.5ex\advance\graphtemp by 1.494in
    \rlap{\kern 5.884in\lower\graphtemp\hbox to 0pt{\hss $\bu$\hss}}%
    \graphtemp=.5ex\advance\graphtemp by 1.132in
    \rlap{\kern 6.065in\lower\graphtemp\hbox to 0pt{\hss $\bu$\hss}}%
    \graphtemp=.5ex\advance\graphtemp by 1.494in
    \rlap{\kern 6.247in\lower\graphtemp\hbox to 0pt{\hss $\bu$\hss}}%
    \graphtemp=.5ex\advance\graphtemp by 1.132in
    \rlap{\kern 6.428in\lower\graphtemp\hbox to 0pt{\hss $\bu$\hss}}%
    \graphtemp=.5ex\advance\graphtemp by 0.769in
    \rlap{\kern 6.428in\lower\graphtemp\hbox to 0pt{\hss $\bu$\hss}}%
    \graphtemp=.5ex\advance\graphtemp by 0.769in
    \rlap{\kern 6.247in\lower\graphtemp\hbox to 0pt{\hss $\bu$\hss}}%
    \graphtemp=.5ex\advance\graphtemp by 0.407in
    \rlap{\kern 6.247in\lower\graphtemp\hbox to 0pt{\hss $\bu$\hss}}%
    \graphtemp=.5ex\advance\graphtemp by 0.407in
    \rlap{\kern 5.884in\lower\graphtemp\hbox to 0pt{\hss $\bu$\hss}}%
    \graphtemp=.5ex\advance\graphtemp by 0.769in
    \rlap{\kern 5.884in\lower\graphtemp\hbox to 0pt{\hss $\bu$\hss}}%
    \special{pa 5703 1132}%
    \special{pa 5884 1494}%
    \special{fp}%
    \special{pa 5884 1494}%
    \special{pa 6247 1494}%
    \special{fp}%
    \special{pa 6247 1494}%
    \special{pa 6428 1132}%
    \special{fp}%
    \special{pa 6428 1132}%
    \special{pa 6428 769}%
    \special{fp}%
    \special{pa 6428 769}%
    \special{pa 6247 407}%
    \special{fp}%
    \special{pa 6247 407}%
    \special{pa 5884 407}%
    \special{fp}%
    \special{pa 5884 407}%
    \special{pa 6247 769}%
    \special{fp}%
    \special{pa 6247 769}%
    \special{pa 6065 1132}%
    \special{fp}%
    \special{pa 6065 1132}%
    \special{pa 5703 1132}%
    \special{fp}%
    \special{pa 6065 1132}%
    \special{pa 5884 1494}%
    \special{fp}%
    \special{pa 6065 1132}%
    \special{pa 6247 1494}%
    \special{fp}%
    \special{pa 6065 1132}%
    \special{pa 6428 1132}%
    \special{fp}%
    \special{pa 5884 407}%
    \special{pa 5884 769}%
    \special{fp}%
    \special{pa 5884 769}%
    \special{pa 6247 769}%
    \special{fp}%
    \special{pa 6247 769}%
    \special{pa 6428 769}%
    \special{fp}%
    \special{pa 6247 407}%
    \special{pa 6247 769}%
    \special{fp}%
    \special{pa 6247 769}%
    \special{pa 6428 1132}%
    \special{fp}%
    \graphtemp=.5ex\advance\graphtemp by 1.132in
    \rlap{\kern 5.613in\lower\graphtemp\hbox to 0pt{\hss $v'$\hss}}%
    \graphtemp=.5ex\advance\graphtemp by 0.878in
    \rlap{\kern 5.884in\lower\graphtemp\hbox to 0pt{\hss $v$\hss}}%
    \graphtemp=.5ex\advance\graphtemp by 0.951in
    \rlap{\kern 6.500in\lower\graphtemp\hbox to 0pt{\hss $e$\hss}}%
    \special{pn 28}%
    \special{pa 5884 769}%
    \special{pa 5884 407}%
    \special{fp}%
    \graphtemp=.5ex\advance\graphtemp by 1.349in
    \rlap{\kern 6.065in\lower\graphtemp\hbox to 0pt{\hss $\bu$\hss}}%
    \graphtemp=.5ex\advance\graphtemp by 1.494in
    \rlap{\kern 6.428in\lower\graphtemp\hbox to 0pt{\hss $\bu$\hss}}%
    \graphtemp=.5ex\advance\graphtemp by 1.856in
    \rlap{\kern 6.428in\lower\graphtemp\hbox to 0pt{\hss $\bu$\hss}}%
    \graphtemp=.5ex\advance\graphtemp by 1.856in
    \rlap{\kern 6.065in\lower\graphtemp\hbox to 0pt{\hss $\bu$\hss}}%
    \special{pn 11}%
    \special{pa 6065 1132}%
    \special{pa 6065 1349}%
    \special{fp}%
    \special{pa 6065 1349}%
    \special{pa 6247 1494}%
    \special{fp}%
    \special{pa 6247 1494}%
    \special{pa 6428 1494}%
    \special{fp}%
    \special{pa 6428 1494}%
    \special{pa 6428 1856}%
    \special{fp}%
    \special{pa 6428 1856}%
    \special{pa 6247 1494}%
    \special{fp}%
    \special{pa 6247 1494}%
    \special{pa 6065 1856}%
    \special{fp}%
    \special{pa 6065 1856}%
    \special{pa 5884 1494}%
    \special{fp}%
    \special{pa 6428 1494}%
    \special{pa 6428 1132}%
    \special{fp}%
    \graphtemp=.5ex\advance\graphtemp by 2.109in
    \rlap{\kern 6.065in\lower\graphtemp\hbox to 0pt{\hss $G'$\hss}}%
    \graphtemp=.5ex\advance\graphtemp by 0.640in
    \rlap{\kern 5.972in\lower\graphtemp\hbox to 0pt{\hss $e_p$\hss}}%
    \graphtemp=.5ex\advance\graphtemp by 0.951in
    \rlap{\kern 5.703in\lower\graphtemp\hbox to 0pt{\hss $\bu$\hss}}%
    \graphtemp=.5ex\advance\graphtemp by 0.769in
    \rlap{\kern 5.522in\lower\graphtemp\hbox to 0pt{\hss $\bu$\hss}}%
    \graphtemp=.5ex\advance\graphtemp by 0.407in
    \rlap{\kern 5.522in\lower\graphtemp\hbox to 0pt{\hss $\bu$\hss}}%
    \graphtemp=.5ex\advance\graphtemp by 0.045in
    \rlap{\kern 5.703in\lower\graphtemp\hbox to 0pt{\hss $\bu$\hss}}%
    \graphtemp=.5ex\advance\graphtemp by 0.045in
    \rlap{\kern 5.341in\lower\graphtemp\hbox to 0pt{\hss $\bu$\hss}}%
    \special{pa 5884 769}%
    \special{pa 5703 951}%
    \special{fp}%
    \special{pa 5703 951}%
    \special{pa 5522 769}%
    \special{fp}%
    \special{pa 5522 769}%
    \special{pa 5884 769}%
    \special{fp}%
    \special{pa 5884 769}%
    \special{pa 5522 407}%
    \special{fp}%
    \special{pa 5522 407}%
    \special{pa 5341 45}%
    \special{fp}%
    \special{pa 5341 45}%
    \special{pa 5703 45}%
    \special{fp}%
    \special{pa 5703 45}%
    \special{pa 5884 407}%
    \special{fp}%
    \special{pa 5884 407}%
    \special{pa 5522 407}%
    \special{fp}%
    \special{pa 5522 769}%
    \special{pa 5884 407}%
    \special{fp}%
    \special{pa 5522 407}%
    \special{pa 5703 45}%
    \special{fp}%
    \graphtemp=.5ex\advance\graphtemp by 1.548in
    \rlap{\kern 5.703in\lower\graphtemp\hbox to 0pt{\hss $H$\hss}}%
    \graphtemp=.5ex\advance\graphtemp by 0.588in
    \rlap{\kern 5.432in\lower\graphtemp\hbox to 0pt{\hss $J$\hss}}%
    \hbox{\vrule depth2.109in width0pt height 0pt}%
    \kern 6.500in
  }%
}%
}

\vspace{-1pc} 
\caption{Reattaching $J$ to form $G'$ from $G$ in
Theorem~\ref{mainthm}.\label{figmain}}
\end{figure}

Let $q=|V(H')|$.  Let $(\VEC vq1)$ be the $2$-simplicial ordering of $H'$ where
$v_q=v$ and $v_1=v'$.  Each vertex of $G$ not in $H'$ belongs to a $2$-tree
that shares exactly one edge with $H'$ and grows from that edge by adding
$2$-simplicial vertices.  Such an edge $e$ cannot be incident to $v$ or $v'$,
since $v$ and $v'$ are simplicial in $G$.  Either $e$ consists of the neighbors
of some $v_j$ when $v_j$ is deleted in $(\VEC vq1)$, or $e$ consists of $v_j$
and one such neighbor.  Let $j^*$ be the largest such index $j$.  Let the
resulting edge $e$ be the {\it crucial edge}, specified as the edge incident to
$v_{j^*}$ rather than the one joining the neighbors of $v_{j^*}$ if both
choices are available.  The example in Figure~\ref{figmain} shows the case 
where $e$ is incident to $v_{j^*}$.  Let $p=q-j^*+1$, and let $w_i=v_{j^*-1+i}$
for $1\le i\le p$.

Let $\hH$ be the $2$-tree obtained from $G$ by deleting $\VEC wp1$ and all the
vertices not in $H'$ that belong to the maximal $2$-tree sharing only the
crucial edge $e$ with $H'$.  In the language of Lemma~\ref{pathineq},
restoring $\VEC wp1$ yields a $2$-tree $G_p$ with the edge $e_p$ being
incident to $v$ and $e$ playing the role of $e'$ or $e_0$.
By Lemma~\ref{pathineq}, $T(G_p;e_p)>T(G_p;e)$.

Now let $H=G_p$, and let $J$ be the $2$-tree induced by the endpoints of $e$
and all the vertices outside $H'$ deleted from $G$ to form $\hH$; note that
$G=H\cup J$.  Let $G'$ be the $2$-tree obtained as $H\cup J$ by using $e_p$ as
the common edge instead of $e$.  Setting $S=\nul$ in Lemma~\ref{anyset} now
yields $T(G')>T(G)$.
\end{proof}

\section{2-Trees with the Fewest Spanning Trees}

We close by proving that the unique $n$-vertex $2$-tree having the fewest
spanning trees is $B_n$, the $2$-tree with $n-2$ simplicial vertices (see
Theorem~\ref{strictbound}).

\begin{theorem}\label{Tmin}
If $G$ is a $n$-vertex $2$-tree other than $B_n$, then $T(G)>T(B_n)$.
\end{theorem}
\begin{proof}
It suffices to show that some $n$-vertex $2$-tree other than $G$ has
fewer spanning trees than $G$.  The book $B_n$ is characterized by the
fact that every two simplicial vertices have the same neighborhood.
If $G\ne B_n$, then let $v_1$ and $v_2$ be simplicial vertices in $G$
whose neighborhoods are not equal.  Let $H=G-\{v_1,v_2\}$.

For $i\in\{1,2\}$, let $e_i$ be the edge joining the neighbors of $v_i$,
let $\beta_i=T(H;e_i)$, and let $G_i$ be the graph obtained from $H$ by adding
two simplicial vertices each having neighborhood $N_G(v_i)$ (see
Figure~\ref{Tmin}).  Let $\gamma=T(H;\{e_1,e_2\})$.

Spanning trees in $G$ have two, three or four edges incident to $\{v_1,v_2\}$;
we compute $T(G)=4T(H)+2\beta_1+2\beta_2+\gamma$.  Spanning trees in $G_i$
have two or three edges not in $H$; we compute $T(G_i)=4T(H)+4\beta_i$.  Thus
\begin{eqnarray*}
2T(G)&=&8T(H)+4\beta_1+4\beta_2+2\gamma\\
T(G_1)+T(G_2)&=&8T(H)+4\beta_1+4\beta_2.
\end{eqnarray*}
Since a spanning tree can be grown from any subtree, $\gamma\ne0$.  Thus
$T(G_1)+T(G_2)<2T(G)$, and $G_1$ or $G_2$ has fewer spanning trees than $G$.
\end{proof}

\vspace{-1pc}
\begin{figure}[h]
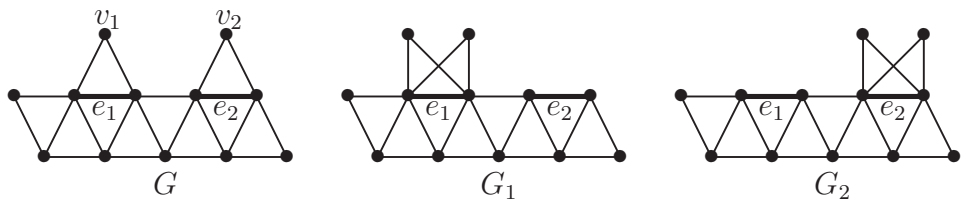

\gpic{
\expandafter\ifx\csname graph\endcsname\relax \csname newbox\endcsname\graph\fi
\expandafter\ifx\csname graphtemp\endcsname\relax \csname newdimen\endcsname\graphtemp\fi
\setbox\graph=\vtop{\vskip 0pt\hbox{%
    \graphtemp=.5ex\advance\graphtemp by 0.730in
    \rlap{\kern 0.198in\lower\graphtemp\hbox to 0pt{\hss $\bu$\hss}}%
    \graphtemp=.5ex\advance\graphtemp by 0.730in
    \rlap{\kern 0.516in\lower\graphtemp\hbox to 0pt{\hss $\bu$\hss}}%
    \graphtemp=.5ex\advance\graphtemp by 0.730in
    \rlap{\kern 0.833in\lower\graphtemp\hbox to 0pt{\hss $\bu$\hss}}%
    \graphtemp=.5ex\advance\graphtemp by 0.730in
    \rlap{\kern 1.151in\lower\graphtemp\hbox to 0pt{\hss $\bu$\hss}}%
    \graphtemp=.5ex\advance\graphtemp by 0.730in
    \rlap{\kern 1.468in\lower\graphtemp\hbox to 0pt{\hss $\bu$\hss}}%
    \special{pn 11}%
    \special{pa 198 730}%
    \special{pa 1468 730}%
    \special{fp}%
    \special{pa 1468 730}%
    \special{pa 1310 413}%
    \special{fp}%
    \special{pa 1310 413}%
    \special{pa 40 413}%
    \special{fp}%
    \special{pa 40 413}%
    \special{pa 198 730}%
    \special{fp}%
    \special{pa 198 730}%
    \special{pa 357 413}%
    \special{fp}%
    \special{pa 357 413}%
    \special{pa 516 730}%
    \special{fp}%
    \special{pa 516 730}%
    \special{pa 675 413}%
    \special{fp}%
    \special{pa 675 413}%
    \special{pa 833 730}%
    \special{fp}%
    \special{pa 833 730}%
    \special{pa 992 413}%
    \special{fp}%
    \special{pa 992 413}%
    \special{pa 1151 730}%
    \special{fp}%
    \special{pa 1151 730}%
    \special{pa 1310 413}%
    \special{fp}%
    \graphtemp=.5ex\advance\graphtemp by 0.413in
    \rlap{\kern 0.040in\lower\graphtemp\hbox to 0pt{\hss $\bu$\hss}}%
    \graphtemp=.5ex\advance\graphtemp by 0.413in
    \rlap{\kern 0.357in\lower\graphtemp\hbox to 0pt{\hss $\bu$\hss}}%
    \graphtemp=.5ex\advance\graphtemp by 0.413in
    \rlap{\kern 0.675in\lower\graphtemp\hbox to 0pt{\hss $\bu$\hss}}%
    \graphtemp=.5ex\advance\graphtemp by 0.413in
    \rlap{\kern 0.992in\lower\graphtemp\hbox to 0pt{\hss $\bu$\hss}}%
    \graphtemp=.5ex\advance\graphtemp by 0.413in
    \rlap{\kern 1.310in\lower\graphtemp\hbox to 0pt{\hss $\bu$\hss}}%
    \graphtemp=.5ex\advance\graphtemp by 0.095in
    \rlap{\kern 0.516in\lower\graphtemp\hbox to 0pt{\hss $\bu$\hss}}%
    \graphtemp=.5ex\advance\graphtemp by 0.095in
    \rlap{\kern 1.151in\lower\graphtemp\hbox to 0pt{\hss $\bu$\hss}}%
    \special{pn 28}%
    \special{pa 357 413}%
    \special{pa 675 413}%
    \special{fp}%
    \special{pa 992 413}%
    \special{pa 1310 413}%
    \special{fp}%
    \special{pn 11}%
    \special{pa 357 413}%
    \special{pa 516 95}%
    \special{fp}%
    \special{pa 516 95}%
    \special{pa 675 413}%
    \special{fp}%
    \special{pa 992 413}%
    \special{pa 1151 95}%
    \special{fp}%
    \special{pa 1151 95}%
    \special{pa 1310 413}%
    \special{fp}%
    \graphtemp=.5ex\advance\graphtemp by 0.476in
    \rlap{\kern 0.516in\lower\graphtemp\hbox to 0pt{\hss $e_1$\hss}}%
    \graphtemp=.5ex\advance\graphtemp by 0.476in
    \rlap{\kern 1.151in\lower\graphtemp\hbox to 0pt{\hss $e_2$\hss}}%
    \graphtemp=.5ex\advance\graphtemp by 0.889in
    \rlap{\kern 0.833in\lower\graphtemp\hbox to 0pt{\hss $G$\hss}}%
    \graphtemp=.5ex\advance\graphtemp by 0.000in
    \rlap{\kern 0.532in\lower\graphtemp\hbox to 0pt{\hss $v_1$\hss}}%
    \graphtemp=.5ex\advance\graphtemp by 0.000in
    \rlap{\kern 1.167in\lower\graphtemp\hbox to 0pt{\hss $v_2$\hss}}%
    \graphtemp=.5ex\advance\graphtemp by 0.730in
    \rlap{\kern 1.944in\lower\graphtemp\hbox to 0pt{\hss $\bu$\hss}}%
    \graphtemp=.5ex\advance\graphtemp by 0.730in
    \rlap{\kern 2.262in\lower\graphtemp\hbox to 0pt{\hss $\bu$\hss}}%
    \graphtemp=.5ex\advance\graphtemp by 0.730in
    \rlap{\kern 2.579in\lower\graphtemp\hbox to 0pt{\hss $\bu$\hss}}%
    \graphtemp=.5ex\advance\graphtemp by 0.730in
    \rlap{\kern 2.897in\lower\graphtemp\hbox to 0pt{\hss $\bu$\hss}}%
    \graphtemp=.5ex\advance\graphtemp by 0.730in
    \rlap{\kern 3.214in\lower\graphtemp\hbox to 0pt{\hss $\bu$\hss}}%
    \special{pa 1944 730}%
    \special{pa 3214 730}%
    \special{fp}%
    \special{pa 3214 730}%
    \special{pa 3056 413}%
    \special{fp}%
    \special{pa 3056 413}%
    \special{pa 1786 413}%
    \special{fp}%
    \special{pa 1786 413}%
    \special{pa 1944 730}%
    \special{fp}%
    \special{pa 1944 730}%
    \special{pa 2103 413}%
    \special{fp}%
    \special{pa 2103 413}%
    \special{pa 2262 730}%
    \special{fp}%
    \special{pa 2262 730}%
    \special{pa 2421 413}%
    \special{fp}%
    \special{pa 2421 413}%
    \special{pa 2579 730}%
    \special{fp}%
    \special{pa 2579 730}%
    \special{pa 2738 413}%
    \special{fp}%
    \special{pa 2738 413}%
    \special{pa 2897 730}%
    \special{fp}%
    \special{pa 2897 730}%
    \special{pa 3056 413}%
    \special{fp}%
    \special{pn 28}%
    \special{pa 2103 413}%
    \special{pa 2421 413}%
    \special{fp}%
    \special{pa 2738 413}%
    \special{pa 3056 413}%
    \special{fp}%
    \graphtemp=.5ex\advance\graphtemp by 0.413in
    \rlap{\kern 1.786in\lower\graphtemp\hbox to 0pt{\hss $\bu$\hss}}%
    \graphtemp=.5ex\advance\graphtemp by 0.413in
    \rlap{\kern 2.103in\lower\graphtemp\hbox to 0pt{\hss $\bu$\hss}}%
    \graphtemp=.5ex\advance\graphtemp by 0.413in
    \rlap{\kern 2.421in\lower\graphtemp\hbox to 0pt{\hss $\bu$\hss}}%
    \graphtemp=.5ex\advance\graphtemp by 0.413in
    \rlap{\kern 2.738in\lower\graphtemp\hbox to 0pt{\hss $\bu$\hss}}%
    \graphtemp=.5ex\advance\graphtemp by 0.413in
    \rlap{\kern 3.056in\lower\graphtemp\hbox to 0pt{\hss $\bu$\hss}}%
    \graphtemp=.5ex\advance\graphtemp by 0.095in
    \rlap{\kern 2.103in\lower\graphtemp\hbox to 0pt{\hss $\bu$\hss}}%
    \graphtemp=.5ex\advance\graphtemp by 0.095in
    \rlap{\kern 2.421in\lower\graphtemp\hbox to 0pt{\hss $\bu$\hss}}%
    \special{pn 11}%
    \special{pa 2103 95}%
    \special{pa 2103 413}%
    \special{fp}%
    \special{pa 2103 413}%
    \special{pa 2421 95}%
    \special{fp}%
    \special{pa 2421 95}%
    \special{pa 2421 413}%
    \special{fp}%
    \special{pa 2421 413}%
    \special{pa 2103 95}%
    \special{fp}%
    \graphtemp=.5ex\advance\graphtemp by 0.476in
    \rlap{\kern 2.262in\lower\graphtemp\hbox to 0pt{\hss $e_1$\hss}}%
    \graphtemp=.5ex\advance\graphtemp by 0.476in
    \rlap{\kern 2.897in\lower\graphtemp\hbox to 0pt{\hss $e_2$\hss}}%
    \graphtemp=.5ex\advance\graphtemp by 0.889in
    \rlap{\kern 2.579in\lower\graphtemp\hbox to 0pt{\hss $G_1$\hss}}%
    \graphtemp=.5ex\advance\graphtemp by 0.730in
    \rlap{\kern 3.690in\lower\graphtemp\hbox to 0pt{\hss $\bu$\hss}}%
    \graphtemp=.5ex\advance\graphtemp by 0.730in
    \rlap{\kern 4.008in\lower\graphtemp\hbox to 0pt{\hss $\bu$\hss}}%
    \graphtemp=.5ex\advance\graphtemp by 0.730in
    \rlap{\kern 4.325in\lower\graphtemp\hbox to 0pt{\hss $\bu$\hss}}%
    \graphtemp=.5ex\advance\graphtemp by 0.730in
    \rlap{\kern 4.643in\lower\graphtemp\hbox to 0pt{\hss $\bu$\hss}}%
    \graphtemp=.5ex\advance\graphtemp by 0.730in
    \rlap{\kern 4.960in\lower\graphtemp\hbox to 0pt{\hss $\bu$\hss}}%
    \special{pa 3690 730}%
    \special{pa 4960 730}%
    \special{fp}%
    \special{pa 4960 730}%
    \special{pa 4802 413}%
    \special{fp}%
    \special{pa 4802 413}%
    \special{pa 3532 413}%
    \special{fp}%
    \special{pa 3532 413}%
    \special{pa 3690 730}%
    \special{fp}%
    \special{pa 3690 730}%
    \special{pa 3849 413}%
    \special{fp}%
    \special{pa 3849 413}%
    \special{pa 4008 730}%
    \special{fp}%
    \special{pa 4008 730}%
    \special{pa 4167 413}%
    \special{fp}%
    \special{pa 4167 413}%
    \special{pa 4325 730}%
    \special{fp}%
    \special{pa 4325 730}%
    \special{pa 4484 413}%
    \special{fp}%
    \special{pa 4484 413}%
    \special{pa 4643 730}%
    \special{fp}%
    \special{pa 4643 730}%
    \special{pa 4802 413}%
    \special{fp}%
    \special{pn 28}%
    \special{pa 3849 413}%
    \special{pa 4167 413}%
    \special{fp}%
    \special{pa 4484 413}%
    \special{pa 4802 413}%
    \special{fp}%
    \graphtemp=.5ex\advance\graphtemp by 0.413in
    \rlap{\kern 3.532in\lower\graphtemp\hbox to 0pt{\hss $\bu$\hss}}%
    \graphtemp=.5ex\advance\graphtemp by 0.413in
    \rlap{\kern 3.849in\lower\graphtemp\hbox to 0pt{\hss $\bu$\hss}}%
    \graphtemp=.5ex\advance\graphtemp by 0.413in
    \rlap{\kern 4.167in\lower\graphtemp\hbox to 0pt{\hss $\bu$\hss}}%
    \graphtemp=.5ex\advance\graphtemp by 0.413in
    \rlap{\kern 4.484in\lower\graphtemp\hbox to 0pt{\hss $\bu$\hss}}%
    \graphtemp=.5ex\advance\graphtemp by 0.413in
    \rlap{\kern 4.802in\lower\graphtemp\hbox to 0pt{\hss $\bu$\hss}}%
    \graphtemp=.5ex\advance\graphtemp by 0.095in
    \rlap{\kern 4.484in\lower\graphtemp\hbox to 0pt{\hss $\bu$\hss}}%
    \graphtemp=.5ex\advance\graphtemp by 0.095in
    \rlap{\kern 4.802in\lower\graphtemp\hbox to 0pt{\hss $\bu$\hss}}%
    \special{pn 11}%
    \special{pa 4484 95}%
    \special{pa 4484 413}%
    \special{fp}%
    \special{pa 4484 413}%
    \special{pa 4802 95}%
    \special{fp}%
    \special{pa 4802 95}%
    \special{pa 4802 413}%
    \special{fp}%
    \special{pa 4802 413}%
    \special{pa 4484 95}%
    \special{fp}%
    \graphtemp=.5ex\advance\graphtemp by 0.476in
    \rlap{\kern 4.008in\lower\graphtemp\hbox to 0pt{\hss $e_1$\hss}}%
    \graphtemp=.5ex\advance\graphtemp by 0.476in
    \rlap{\kern 4.643in\lower\graphtemp\hbox to 0pt{\hss $e_2$\hss}}%
    \graphtemp=.5ex\advance\graphtemp by 0.889in
    \rlap{\kern 4.325in\lower\graphtemp\hbox to 0pt{\hss $G_2$\hss}}%
    \hbox{\vrule depth0.952in width0pt height 0pt}%
    \kern 5.000in
  }%
}%
}

\vspace{-1pc} 
\caption{Graphs $G_1$ and $G_2$ compared with $G$ in 
Theorem~\ref{Tmin}.\label{figTmin}}
\end{figure}


\end{document}